\documentclass[letterpaper]{article}
\usepackage{amsmath,amssymb,amsthm}
\newtheorem{definition}{Definition}
\newtheorem{theorem}{Theorem}
\newtheorem{lemma}{Lemma}
\newtheorem{subroutine}{Subroutine}
\newtheorem{property}{Property}
\usepackage{xargs}
\usepackage{xcolor}
\definecolor{rrrrrr}{rgb}{0.75,0,0}
\definecolor{gggggg}{rgb}{0,0.75,0}
\definecolor{bbbbbb}{rgb}{0,0,0.75}
\usepackage{tikz}
\usetikzlibrary{arrows,intersections}
\usepackage{todonotes}
\newcommandx{\fix}[1]{\todo[inline,caption={}]{Fix: #1}}
\newcommandx{\missing}[1]{\todo[inline,caption={}]{Missing: #1}}
\usepackage{enumerate}
\title{Computing the Yolk in Spatial Voting Games without Computing Median Lines}
\author{
	Joachim Gudmundsson and Sampson Wong \\
	University of Sydney \\ 
	Sydney, Australia \\
	joachim.gudmundsson@sydney.edu.au, swon7907@sydney.edu.au
}
\begin{document}
\maketitle
\begin{abstract}
The yolk is an important concept in spatial voting games: the yolk center generalises the equilibrium and the yolk radius bounds the uncovered set. We present near-linear time algorithms for computing the yolk in the plane. To the best of our knowledge our algorithm is the first that does not precompute median lines, and hence is able to break  the best known upper bound of~$O(n^{4/3})$ on the number of limiting median lines. We avoid this requirement by carefully applying Megiddo's parametric search technique, which is a powerful framework that could lead to faster algorithms for other spatial voting problems.
\end{abstract}

\section{Introduction}
\label{sec:introduction}

Voting theory is concerned with preference aggregation and group decision making. A classic framework for aggregating voter's preferences is the Downsian~\cite{downs1957economic}, or spatial model of voting. 

In this model, voters are positioned on a `left-right' continuum along multiple ideological dimensions, such as economic, social or religious. These dimensions together form the policy space. Each voter is required to choose a single candidate from a set of candidates, and a common voter preference function is a metric/distance function within the policy space. An intuitive reason behind using metric preferences is that voters tend to prefer candidates ideologically similar to themselves.

The spatial model of voting with metric preferences have been studied extensively, both theoretically~\cite{MCKELVEY1976472,mckelvey1986covering,enelow1984spatial,miller1989geometry,tovey1993some} and empirically~\cite{poole1984polarization,poole1991patterns,poole2001d,ordeshook1993spatial,schofield2003critical,schofield2004equilibrium,schofield2007spatial}. Recently, lower bounds were provided on the distortion of voting rules in the spatial model, and interestingly, metrics other than the Euclidean metric were considered~\cite{DBLP:conf/aaai/AnshelevichBP15,DBLP:conf/aaai/SkowronE17,DBLP:conf/sigecom/GoelKM17}.

We focus our attention on two-candidate spatial voting games, where the winner is the candidate preferred by a simple majority of voters. In a one dimension policy space, Black's Median Voter Theorem~\cite{black1948rationale} states that a voting equilibrium (alt. Condorcet winner, plurality point, pure Nash equilibrium) is guaranteed to exist and coincides with the median voter. 

Naturally, social choice theorists searched for the equilibrium in the two dimensional policy space, but these attempts were shown to be fruitless \cite{plott1967notion}. The initial reaction was one of cynicism \cite{MCKELVEY1976472}, but in response a multitude of generalisations were developed, with the yolk being one such concept~\cite{mckelvey1986covering,miller1989geometry}. The yolk in the Euclidean~$\mathcal L_2$ metric is defined as the minimum radius disk that intersects all median lines of the voters.

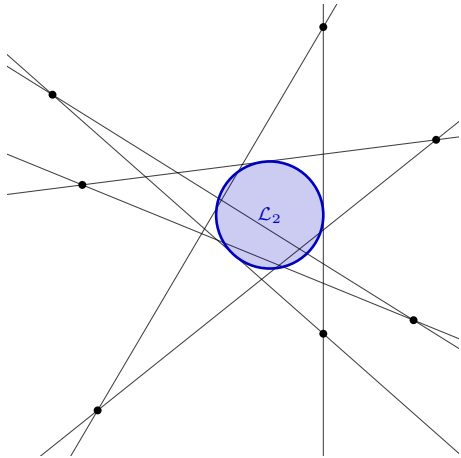
\begin{figure}[ht]
	\centering
	\begin{tikzpicture}[line cap=round,line join=round,>=triangle 45,x=0.6cm,y=0.6cm]
	\clip(-5,-5) rectangle (5,5);

	\begin{scriptsize}
	\fill [color=black] (-4,3) circle (1.5pt);
	\fill [color=black] (-3,-4) circle (1.5pt);
	\fill [color=black] (2,4.5) circle (1.5pt);
	\fill [color=black] (2,-2.3) circle (1.5pt);
	\fill [color=black] (-3.34,1) circle (1.5pt);
	\fill [color=black] (4,-2) circle (1.5pt);
	\fill [color=black] (4.5,2) circle (1.5pt);
	\draw [opacity=0.75, shorten >=-10cm, shorten <=-10cm]  (-4, 3)--(2, -2.3) ;
	\draw [opacity=0.75, shorten >=-10cm, shorten <=-10cm]  (-4, 3)--(4, -2) ;
	\draw [opacity=0.75, shorten >=-10cm, shorten <=-10cm]  (-3, -4)--(2, 4.5) ;
	\draw [opacity=0.75, shorten >=-10cm, shorten <=-10cm]  (-3, -4)--(4.5, 2) ;
	\draw [opacity=0.75, shorten >=-10cm, shorten <=-10cm]  (2, 4.5)--(2, -2.3) ;
	\draw [opacity=0.75, shorten >=-10cm, shorten <=-10cm]  (-3.34, 1)--(4, -2) ;
	\draw [opacity=0.75, shorten >=-10cm, shorten <=-10cm]  (-3.34, 1)--(4.5, 2) ;
	\filldraw [line width=1pt, color=bbbbbb, fill opacity=0.2] (0.81275940334137398, 0.3328293029584386) circle (1.18724059666);
	\draw[color=bbbbbb] (0.81275940334137398, 0.3328293029584386) node {$\mathcal L_2$};
	\end{scriptsize}
	\end{tikzpicture}
	\caption{The~$\mathcal L_2$ yolk intersects all median lines of voters.}
	\label{fig:l2_yolk}
\end{figure}

The yolk is an important concept in spatial voting games due to its simplicity and its relationship to other concepts. The yolk radius provides approximate bounds on the uncovered set~\cite{feld1987uncovered,miller1980new,miller1989geometry}, limits on agenda control~\cite{feld1989limits}, Shapley-Owen power scores~\cite{feld1990theorem}, the Finagle point~\cite{wuffle1989finagle} and the~$\varepsilon$-core~\cite{tovey2011finagle}. As such, studies on the size of the yolk~\cite{feld1988centripetal,koehler1990size,DBLP:journals/mss/Tovey10a} translate to these other concepts as well.

From the perspective of computational social choice, this raises the following problem: Are there efficient algorithms for computing the yolk? Fast algorithms would, for instance, facilitate empirical studies on large data sets. Tovey~\cite{tovey1992polynomial} provides the first polynomial time algorithm, which in two dimensions, computes the yolk in~$O(n^{4})$ time. De Berg et al.~\cite{DBLP:conf/compgeom/BergGM16} provides an improved~$O(n^{4/3} \log^{1+\varepsilon} n)$ time algorithm
for the same.

The shortcoming of existing algorithms is that they require the computation of all limiting median lines, which are median lines that pass through at least two voters \cite{stone1992limiting}. However, there are~$\Omega(n e^{\sqrt{\log n}})$~\cite{DBLP:conf/compgeom/Toth00} limiting median lines in the worst case. Moreover, the best known upper bound of~$O(n^{4/3})$ seems difficult to improve on~\cite{DBLP:conf/focs/Dey97}. It is an open problem as to whether there is a faster algorithm that computes the yolk without precomputing all limiting median lines .

\subsection{Problem Statement}
Given a set~$V$ of~$n$ points in the plane, a median line of~$V$ is any line that divides the plane into two closed halfplanes, each with at most~$n/2$ points. The yolk is a minimum radius disk in the~$\mathcal L_p$ metric that intersects all median lines of~$V$.

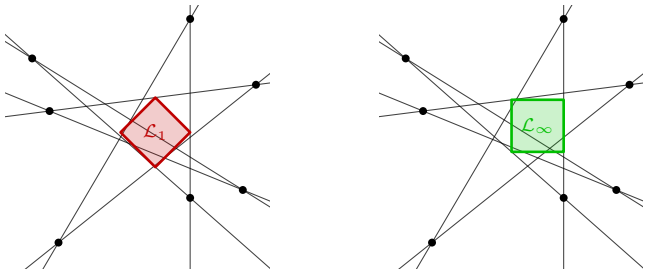
\begin{figure}[ht]
	\centering
	\begin{minipage}{0.4\textwidth}
		\begin{tikzpicture}[line cap=round,line join=round,>=triangle 45,x=0.35cm,y=0.35cm]
		\clip(-5,-5) rectangle (5,5);
		\begin{scriptsize}
		\fill [color=black] (-4,3) circle (1.5pt);
		\fill [color=black] (-3,-4) circle (1.5pt);
		\fill [color=black] (2,4.5) circle (1.5pt);
		\fill [color=black] (2,-2.3) circle (1.5pt);
		\fill [color=black] (-3.34,1) circle (1.5pt);
		\fill [color=black] (4,-2) circle (1.5pt);
		\fill [color=black] (4.5,2) circle (1.5pt);
		\draw [opacity=0.75, shorten >=-10cm, shorten <=-10cm]  (-4, 3)--(2, -2.3) ;
		\draw [opacity=0.75, shorten >=-10cm, shorten <=-10cm]  (-4, 3)--(4, -2) ;
		\draw [opacity=0.75, shorten >=-10cm, shorten <=-10cm]  (-3, -4)--(2, 4.5) ;
		\draw [opacity=0.75, shorten >=-10cm, shorten <=-10cm]  (-3, -4)--(4.5, 2) ;
		\draw [opacity=0.75, shorten >=-10cm, shorten <=-10cm]  (2, 4.5)--(2, -2.3) ;
		\draw [opacity=0.75, shorten >=-10cm, shorten <=-10cm]  (-3.34, 1)--(4, -2) ;
		\draw [opacity=0.75, shorten >=-10cm, shorten <=-10cm]  (-3.34, 1)--(4.5, 2) ;
		\filldraw [line width=1pt, color=rrrrrr, fill opacity=0.2] (1.9999999999999996, 0.19022537279710786)--(0.67775643922277429, 1.5124689335743331)--(-0.64448712155445087, 0.19022537279710786)--(0.67775643922277429, -1.1320181879801172)--cycle;
		\draw[color=rrrrrr] (0.67775643922277429, 0.19022537279710786) node {$\mathcal L_1$};
		\end{scriptsize}
		\end{tikzpicture}
	\end{minipage}
	\begin{minipage}{0.4\textwidth}
		\begin{tikzpicture}[line cap=round,line join=round,>=triangle 45,x=0.35cm,y=0.35cm]
\clip(-5,-5) rectangle (5,5);
		\begin{scriptsize}
		\fill [color=black] (-4,3) circle (1.5pt);
		\fill [color=black] (-3,-4) circle (1.5pt);
		\fill [color=black] (2,4.5) circle (1.5pt);
		\fill [color=black] (2,-2.3) circle (1.5pt);
		\fill [color=black] (-3.34,1) circle (1.5pt);
		\fill [color=black] (4,-2) circle (1.5pt);
		\fill [color=black] (4.5,2) circle (1.5pt);
		\draw [opacity=0.75, shorten >=-10cm, shorten <=-10cm]  (-4, 3)--(2, -2.3) ;
		\draw [opacity=0.75, shorten >=-10cm, shorten <=-10cm]  (-4, 3)--(4, -2) ;
		\draw [opacity=0.75, shorten >=-10cm, shorten <=-10cm]  (-3, -4)--(2, 4.5) ;
		\draw [opacity=0.75, shorten >=-10cm, shorten <=-10cm]  (-3, -4)--(4.5, 2) ;
		\draw [opacity=0.75, shorten >=-10cm, shorten <=-10cm]  (2, 4.5)--(2, -2.3) ;
		\draw [opacity=0.75, shorten >=-10cm, shorten <=-10cm]  (-3.34, 1)--(4, -2) ;
		\draw [opacity=0.75, shorten >=-10cm, shorten <=-10cm]  (-3.34, 1)--(4.5, 2) ;
		\filldraw [line width=1pt, color=gggggg, fill opacity=0.2] (1.9999999999999996, 1.4285986129905268)--(1.9999999999999996, -0.55118826116373498)--(0.020213125845737756, -0.55118826116373498)--(0.020213125845737756, 1.4285986129905268)--cycle;
		\draw[color=gggggg] (1.01010656292, 0.43870517591) node {$\mathcal L_\infty$};
		\end{scriptsize}
		\end{tikzpicture}
	\end{minipage}
	\caption{Example of yolks in the~$\mathcal L_1$ and~$\mathcal L_\infty$ metrics.}
	\label{fig:l1_yolk}
\end{figure}

We compute yolks in the~$\mathcal L_1$ (Taxicab), the~$\mathcal L_2$ (Euclidean), and the~$\mathcal L_\infty$ (Uniform) metrics. As shown in Figure~\ref{fig:l1_yolk}, the yolk in~$\mathcal L_1$ is the smallest~$45^\circ$-rotated square and in~$\mathcal L_\infty$ the smallest axis-parallel square, that intersects all median lines of~$V$. 

\subsection{Our Contributions and Results}
Our contributions are, first, an algorithm that computes the yolk in the~$\mathcal L_1$ and~$\mathcal L_\infty$ metrics in~$O(n \log^7 n)$ time, and second, an algorithm that computes a~$(1+\varepsilon)$-approximation of the yolk in the~$\mathcal L_2$ metric in~$O(n \log^7 n \cdot \log^4 \frac 1 \varepsilon)$ time.

We achieve the improved upper bounds by carefully applying Megiddo's~\cite{megiddo1983applying} parametric search technique, which is a powerful yet complex technique and that could be useful for other spatial voting problems.

The parametric search technique is a framework for converting decision algorithms into optimisation algorithms. For the yolk problem, a decision algorithm would decide whether a given disk intersects all median lines. If this decision algorithm satisifies the three properties as specified by the framework, then Megiddo's result states that there is an efficient algorithm to compute the yolk. 

For the purposes of designing a decision algorithm with the desired properties, we instead consider the more general problem of finding the smallest regular,~$k$-sided polygon that intersects all median lines of~$V$. The regular~$k$-sided polygon~$P_k(r,x,y)$ is shown in Figure~\ref{fig:P_k} and is defined as:

\begin{definition}
Given an integer~$k \geq 3$, construct the regular~$k$-sided polygon~$P_k(r,x,y)$ by:
\begin{itemize}
	\item Constructing a circle with radius~$r$ and centered at~$(x,y)$.
	\item Placing a vertex at the top-most point on the circle, i.e. at~$(x,y+r)$.
	\item Placing the remaining~$k-1$ vertices around the circle so that the~$k$ vertices are evenly spaced.
\end{itemize}
\end{definition}

\begin{figure}[ht]
	\centering
	\begin{tikzpicture}[line cap=round,line join=round,>=triangle 45,x=1.0cm,y=1.0cm,scale=0.6]
	\clip(-4,-4) rectangle (4,4);
	\fill[line width=2.pt,color=rrrrrr,fill=rrrrrr,fill opacity=0.10000000149011612] (0.,3.) -- (-1.5,2.598076211353316) -- (-2.598076211353316,1.5) -- (-3.,0.) -- (-2.598076211353316,-1.5) -- (-1.5,-2.5980762113533156) -- (0.,-3.) -- (1.5,-2.5980762113533165) -- (2.5980762113533156,-1.5) -- (3.,0.) -- (2.598076211353317,1.5) -- (1.5,2.598076211353315) -- cycle;
	\draw [line width=2.pt,color=rrrrrr] (0.,3.)-- (-1.5,2.598076211353316);
	\draw [line width=2.pt,color=rrrrrr] (-1.5,2.598076211353316)-- (-2.598076211353316,1.5);
	\draw [line width=2.pt,color=rrrrrr] (-2.598076211353316,1.5)-- (-3.,0.);
	\draw [line width=2.pt,color=rrrrrr] (-3.,0.)-- (-2.598076211353316,-1.5);
	\draw [line width=2.pt,color=rrrrrr] (-2.598076211353316,-1.5)-- (-1.5,-2.5980762113533156);
	\draw [line width=2.pt,color=rrrrrr] (-1.5,-2.5980762113533156)-- (0.,-3.);
	\draw [line width=2.pt,color=rrrrrr] (0.,-3.)-- (1.5,-2.5980762113533165);
	\draw [line width=2.pt,color=rrrrrr] (1.5,-2.5980762113533165)-- (2.5980762113533156,-1.5);
	\draw [line width=2.pt,color=rrrrrr] (2.5980762113533156,-1.5)-- (3.,0.);
	\draw [line width=2.pt,color=rrrrrr] (3.,0.)-- (2.598076211353317,1.5);
	\draw [line width=2.pt,color=rrrrrr] (2.598076211353317,1.5)-- (1.5,2.598076211353315);
	\draw [line width=2.pt,color=rrrrrr] (1.5,2.598076211353315)-- (0.,3.);
	\begin{scriptsize}
	\draw [fill=black] (0.,0.) circle (2.0pt);
	\draw[color=black] (0,0.4) node {$(x,y)$};
	\draw [fill=black] (0.,3.) circle (2.0pt);
	\draw[color=black] (0.1,3.4) node {$(x,y+r)$};
	\draw[color=rrrrrr] (0,-1) node {$P_k(r,x,y)$};
	\end{scriptsize}
	\end{tikzpicture}
	\caption{The regular,~$k$-sided polygon~$P_k(r,x,y)$.}
	\label{fig:P_k}
\end{figure}
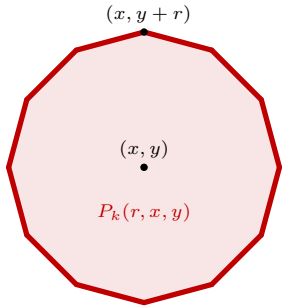

In Section~\ref{sec:decision_algorithm}, we present the decision algorithm, which given a regular,~$k$-sided polygon~$P_k(r,x,y)$, decides whether the polygon intersects all median lines of~$V$. Next, in Section~\ref{sec:parametric_search}, we apply Megiddo's technique to the decision algorithm and prove the convexity and parallelisability properties. This leaves one final property, the existence of critical hyperplanes, left to check. We prove this final property in Sections~4-6, thus completing the parametric search. Finally, in Section~7, we show that our general problem for the regular,~$k$-sided polygon~$P_k(r,x,y)$ implies the claimed running times by setting~$k=4$ for~$\mathcal L_1$ and~$\mathcal L_\infty$, and~$k = \frac 1 \varepsilon$ for~$\mathcal L_2$.

\section{Decision Algorithm}
\label{sec:decision_algorithm}

The aim of this section is to design an algorithm that solves the following decision problem: 

\begin{definition}
Given an integer~$k \geq 3$ and a set~$V$ of~$n$ points in the plane, the decision problem~$D_{k,V}(r,x,y)$ is to decide whether the polygon~$P_k(r,x,y)$ intersects all median lines of~$V$.
\end{definition}

We show that there is a comparison-based decision algorithm that solves~$D_{k,V}(r,x,y)$ in~$O(n \log n \cdot \log k)$ time, provided the following two comparison-based subroutines. 

\begin{subroutine}
A comparison-based subroutine that, given a point~$p$ and a regular~$k$-sided polygon~$P_k(r,x,y)$, decides if~$p$ is outside~$P_k(r,x,y)$ in~$O(\log k)$ time.
\end{subroutine}

\begin{subroutine}
A comparison-based subroutine that, given points~$p, q$ outside a regular~$k$-sided polygon~$P_k(r,x,y)$, computes the relative clockwise order of the four tangent lines drawn from~$\{p,q\}$ to~$P_k(r,x,y)$ in~$O(\log k)$ time.
\end{subroutine}

Although the running time of these two subroutines are not too difficult to prove, we shall see in Section~3 that these subroutines must satisfy a stronger requirement for the parametric search technique to apply. We will formally define the stronger requirement in the next section. To avoid repetition, we simultaneously address the subroutine and the stronger requirement in Sections~5 and~6. But for now, we assume the subroutines exist and present the decision algorithm:

\begin{theorem}
\label{theorem:decide_pk}
Given an integer~$k \geq 3$ and a set~$V$ of~$n$ points in the plane, there is a comparison-based algorithm that solve the decision problem~$D_{k,V}(r,x,y)$ in~$O(n \log n \cdot \log k)$ time, provided that Subroutine~1 and Subroutine~2 exist.
\end{theorem}

\begin{proof}
The proof comes in three parts. First, we transform the decision problem~$D_{k,V}(r,x,y)$ into an equivalent form that does not have median lines in its statement. Then, we present a sweep line algorithm for the transformed version. Finally, we perform an analysis of the running time.

Consider for now a single median line~$m_g$ that has gradient~$g$. Construct two parallel lines~$t_U(g)$ and~$t_D(g)$ that also have gradient~$g$, but are tangent to~$P_k(r,x,y)$ from above and below respectively. If the median line~$m_g$ intersects~$P_k(r,x,y)$, as shown in Figure~\ref{fig:t_U}, then~$m_g$ must be in between~$t_U(g)$ and~$t_D(g)$.

\begin{figure}[ht]
	\centering
	\begin{tikzpicture}[line cap=round,line join=round,>=triangle 45,x=1.0cm,y=1.0cm,scale=0.6]
	\clip(-5,-4) rectangle (5,4);
	\fill[line width=2.pt,color=rrrrrr,fill=rrrrrr,fill opacity=0.10000000149011612] (0.,3.) -- (-1.5,2.598076211353316) -- (-2.598076211353316,1.5) -- (-3.,0.) -- (-2.598076211353316,-1.5) -- (-1.5,-2.5980762113533156) -- (0.,-3.) -- (1.5,-2.5980762113533165) -- (2.5980762113533156,-1.5) -- (3.,0.) -- (2.598076211353317,1.5) -- (1.5,2.598076211353315) -- cycle;
	\draw [line width=2.pt,color=rrrrrr] (0.,3.)-- (-1.5,2.598076211353316);
	\draw [line width=2.pt,color=rrrrrr] (-1.5,2.598076211353316)-- (-2.598076211353316,1.5);
	\draw [line width=2.pt,color=rrrrrr] (-2.598076211353316,1.5)-- (-3.,0.);
	\draw [line width=2.pt,color=rrrrrr] (-3.,0.)-- (-2.598076211353316,-1.5);
	\draw [line width=2.pt,color=rrrrrr] (-2.598076211353316,-1.5)-- (-1.5,-2.5980762113533156);
	\draw [line width=2.pt,color=rrrrrr] (-1.5,-2.5980762113533156)-- (0.,-3.);
	\draw [line width=2.pt,color=rrrrrr] (0.,-3.)-- (1.5,-2.5980762113533165);
	\draw [line width=2.pt,color=rrrrrr] (1.5,-2.5980762113533165)-- (2.5980762113533156,-1.5);
	\draw [line width=2.pt,color=rrrrrr] (2.5980762113533156,-1.5)-- (3.,0.);
	\draw [line width=2.pt,color=rrrrrr] (3.,0.)-- (2.598076211353317,1.5);
	\draw [line width=2.pt,color=rrrrrr] (2.598076211353317,1.5)-- (1.5,2.598076211353315);
	\draw [line width=2.pt,color=rrrrrr] (1.5,2.598076211353315)-- (0.,3.);
	\draw [line width=1.pt,domain=-7.73603361203088:32.919135820740756] plot(\x,{(-7.402020131495246--2.336227153423486*\x)/3.636807218215946});
	\draw [line width=1.pt,domain=-7.73603361203088:32.919135820740756,color=black] plot(\x,{(--15.618408410345495--2.8169558059296755*\x)/4.385158007168881});
	\draw [line width=1.pt,domain=-7.73603361203088:32.919135820740756,color=black] plot(\x,{(-15.618408410345486--2.8169558059296738*\x)/4.3851580071688785});
	\begin{scriptsize}
	\draw[color=rrrrrr] (0,0) node {$P_k(r,x,y)$};
	\draw[color=black] (4,1) node {$m_g$};
	\draw[color=black] (1.4,3.8) node {$t_U(g)$};
	\draw[color=black] (4,-1.7) node {$t_D(g)$};
	\end{scriptsize}
	\end{tikzpicture}
	\caption{The relative positions of~$m_g$,~$t_U(g)$ and~$t_D(g)$ if~$m_g$ intersects the $k$-sided regular polygon~$P_k(r,x,y)$.}
	\label{fig:t_U}
\end{figure}
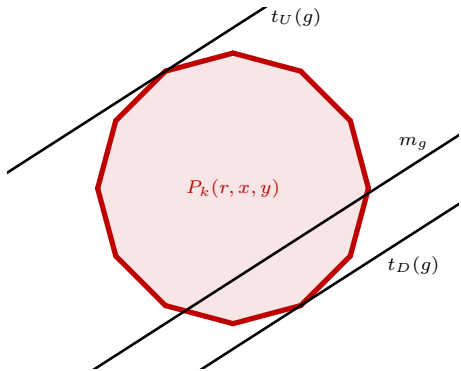

We will decide whether all median lines of gradient~$g$ are between~$t_U(g)$ and~$t_D(g)$, as this would immediately decide whether all median lines of gradient~$g$ intersects~$P_k(r,x,y)$. We will solve this restricted decision problem by counting the number of points in~$V$ above~$t_U(g)$ and the number of points in~$V$ below~$t_D(g)$. 

Let~$t_U^+(g)$ be the number of points in~$V$ that are above~$t_U(g)$, and similarly~$t_D^-(g)$ for the points in~$V$ below~$t_D(g)$. Suppose that~$t_U^+(g) < n/2$ and~$t_D^-(g) < n/2$. Then there cannot be a median line of gradient~$g$ above~$t_U(g)$ or below~$t_D(g)$, since one side of the median line, in particular the side that contains the polygon, will have more than $n/2$ points. Hence, if~$t_U^+(g) < n/2$ and~$t_D^-(g) < n/2$, then all median lines of gradient~$g$ must be between~$t_U(g)$ and~$t_D(g)$.

Conversely, suppose that all median lines of gradient~$g$ are between~$t_U(g)$ and~$t_D(g)$. Then if~$t_U^+(g) \geq n/2$, we can move~$t_U(g)$ continuously upwards until it becomes a median line, which is a contradiction. So in this case, we know~$t_U^+(g) < n/2$ and~$t_D^-(g) < n/2$. 

In summary, we have transformed the decision problem into one that does not have median lines in its statement: All median lines intersect~$P_k(r,x,y)$ if for all gradients~$g$, the pair of inequalities~$t_U^+(g) < n/2$ and~$t_D^-(g) < n/2$ hold.

We present a sweep line algorithm that computes whether the pair of inequalities hold for all gradients~$g$. Let~$t$ be an arbitrary line tangent to the polygon~$P_k(r,x,y)$, and define~$t^+$ to be the open halfplane that has~$t$ as its boundary and does not include the polygon~$P_k(r,x,y)$. Then all median lines intersect~$P_k(r,x,y)$ if and only if for all positions of~$t$, the open halfplane~$t^+$ contains less than~$n/2$ points.

\begin{figure}[ht]
	\centering
	\begin{tikzpicture}[line cap=round,line join=round,>=triangle 45,x=1.0cm,y=1.0cm,scale=0.6]
	\clip(-5,-4) rectangle (5,5);
	\fill[line width=2.pt,color=rrrrrr,fill=rrrrrr,fill opacity=0.10000000149011612] (0.,3.) -- (-1.5,2.598076211353316) -- (-2.598076211353316,1.5) -- (-3.,0.) -- (-2.598076211353316,-1.5) -- (-1.5,-2.5980762113533156) -- (0.,-3.) -- (1.5,-2.5980762113533165) -- (2.5980762113533156,-1.5) -- (3.,0.) -- (2.598076211353317,1.5) -- (1.5,2.598076211353315) -- cycle;
	\draw [line width=2.pt,color=rrrrrr] (0.,3.)-- (-1.5,2.598076211353316);
	\draw [line width=2.pt,color=rrrrrr] (-1.5,2.598076211353316)-- (-2.598076211353316,1.5);
	\draw [line width=2.pt,color=rrrrrr] (-2.598076211353316,1.5)-- (-3.,0.);
	\draw [line width=2.pt,color=rrrrrr] (-3.,0.)-- (-2.598076211353316,-1.5);
	\draw [line width=2.pt,color=rrrrrr] (-2.598076211353316,-1.5)-- (-1.5,-2.5980762113533156);
	\draw [line width=2.pt,color=rrrrrr] (-1.5,-2.5980762113533156)-- (0.,-3.);
	\draw [line width=2.pt,color=rrrrrr] (0.,-3.)-- (1.5,-2.5980762113533165);
	\draw [line width=2.pt,color=rrrrrr] (1.5,-2.5980762113533165)-- (2.5980762113533156,-1.5);
	\draw [line width=2.pt,color=rrrrrr] (2.5980762113533156,-1.5)-- (3.,0.);
	\draw [line width=2.pt,color=rrrrrr] (3.,0.)-- (2.598076211353317,1.5);
	\draw [line width=2.pt,color=rrrrrr] (2.598076211353317,1.5)-- (1.5,2.598076211353315);
	\draw [line width=2.pt,color=rrrrrr] (1.5,2.598076211353315)-- (0.,3.);
	\draw [line width=1.pt,domain=-5:5,color=black] plot(\x,{3});
	\fill[line width=2.pt,color=black,fill=black,fill opacity=0.1] (-5.,5.) -- (5.,5.) -- (5.,3) -- (-5.,3) -- cycle;
	\begin{scriptsize}
	\draw[color=rrrrrr] (0,0) node {$P_k(r,x,y)$};
	\draw[color=black] (4.5,3.3) node {$t$};
	\draw[color=black] (0.3,4.3) node {$t^+$};
	\end{scriptsize}
	\end{tikzpicture}
	\caption{The rotating sweepline~$t$ and the open halfplane~$t^+$.}
	\label{fig:t+}
\end{figure}
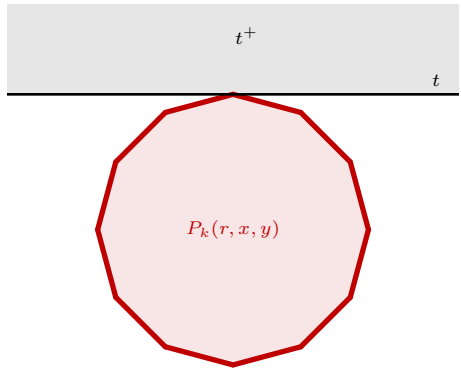

The tangent line~$t$ is a clockwise rotating sweep line and the invariant maintained by the sweep line algorithm is the number of points of~$V$ inside the region~$t^+$. Take any tangent line~$t_0$ to be the starting line, and calculate the number of points in~$t_0^+$. From here, define an event to be when the line~$t$ passes through a point. There are two events for each point outside~$P_k(r,x,y)$; there is one event for when the point enters the region~$t^+$, and one for when it exits. There are no events for points of~$V$ that lie inside~$P_k(r,x,y)$. The unsorted set of event points can be computed by applying Subroutine~1 to each of point in~$V$. 

We sort the set of event points in a clockwise fashion. If we consider only two voters, their associated events can be sorted using Subroutine~2. We can extend this to sort the associated events of all voters with any standard comparison-based sorting algorithm, for example Merge sort.

Once the sorted set of events is computed, we process the events in order. At each new event we maintain our invariant, the number of points inside the region~$t^+$. This value increases by one at ``entry'' events and decreases by one at ``exit'' events. Finally, we return whether our invariant remained less than $n/2$ at all events.

The running time analysis for the algorithm is as follows. Computing the points outside~$P_k(r,x,y)$ takes~$O(\log k)$ time per point by Subroutine~1, so in total this takes~$O(n \log k)$ time. Computing the sorted order of the event points takes~$O(\log k)$ time per comparison by Subroutine~2, which adds up to~$O(n \log n \cdot \log k)$ time. Processing the sorted event points takes~$O(n)$ time. Adding these gives the stated bound.
\end{proof}

\section{Parametric Search}
\label{sec:parametric_search}

Parametric search is a powerful yet complex technique for solving optimisation problems. The two steps involved in this technique are, first, to design a decision algorithm, and second, to convert the decision algorithm into an optimisation algorithm.

For example, our parameter space is $(r,x,y) \in \mathbb R^3$, our decision algorithm is stated in Theorem~1, and our optimisation objective is to minimise~$r \in \mathbb R^+$.

\subsection{Preliminaries}

Megiddo's~(1983) states the requirements for converting the decision algorithm into an optimisation algorithm. First, let us introduce some notation. Let~$\mathbb R^d$ be a parameter space, let~$\lambda \in \mathbb R^d$ be a parameter and let~$D(\lambda)$ be a decision problem that either evaluates to true or false. Then the first requirement is for the decision problem $D(\lambda)$.

\begin{property}
The set of parameters $\{\lambda \in \mathbb R^d: D(\lambda)\}$ that satisfies the decision problem is convex.
\end{property}

Convexity guarantees that the optimisation algorithm finds the global optimum. 

The second property of the technique relates to the decision algorithm. Let~$\mathcal A(\lambda)$ be a comparison-based decision algorithm that computes $D(\lambda)$. Let $C(\lambda)$ be any comparison in the comparison-based decision algorithm~$\mathcal A(\lambda)$. The comparison $C(\lambda)$ is said to have an associated critical hyperplane in $\mathbb R^d$ if the result of the comparison is linearly separable with respect to $\lambda \in \mathbb R^d$. Formally, suppose that the comparison $C(\lambda)$ evaluates to either $>$, $=$ or $<$. Then we say that the $(d-1)$-dimensional hyperplane $H \subset \mathbb R^d$ is the associated critical hyperplane of $C(\lambda)$ if $C$ evaluates to $>$, $=$ or $<$ if and only if $\lambda$ is above, on, or below $H$ respectively. The comparisons of the decision algorithm must satisfy the following property.

\begin{property}
Every comparison $C(\lambda)$ in the comparison-based decision algorithm $\mathcal A(\lambda)$ either (i) does not depend on $\lambda$, or (ii) has an associated critical hyperplane in $\mathbb R^d$.
\end{property}

This requirement allows us to compute a large set of critical hyperplanes that determines the result of $\mathcal A(\lambda)$. Moreover, the optimum must lie on one of these critical hyperplanes, since the result of $\mathcal A(\lambda)$ locally changes sign at the optimum. The new search space now has dimension $d-1$ instead of dimension $d$, and we can recursively apply this procedure to reduce the dimension further. For details see~\cite{agarwal1998efficient}. 

The final property speeds up the parametric search.

\begin{property}
The decision algorithm has an efficient parallel algorithm.
\end{property}

If the decision algorithm~$\mathcal A(\lambda)$ runs in~$T_s$ time and runs on~$P$ processors in~$T_p$ parallel steps, then the parametric search over $\lambda \in \mathbb R^d$ runs in~$O(T_p P + T_s (T_p \log P)^d)$ time~\cite{agarwal1998efficient}.

\subsection{Applying the technique}

To apply the parametric search technique, we show that our decision problem~$D_{k,V}(r,x,y)$ satisfies Properties~1-3.

\begin{lemma}
Given an integer~$k \geq 3$ and a set~$V$ of~$n$ points in the plane, the set of parameters $\{(r,x,y): D_{k,V}(r,x,y)\}$ that satisfies the decision problem is convex.
\end{lemma}

\begin{proof}
Suppose we are given a convex combination~$\lambda_3 = \alpha \lambda_1 + (1-\alpha) \lambda_2$ of the two parameters~$\lambda_1, \lambda_2 \in \mathbb R^3$. Then the polygon~$P_k(\lambda_3)$ is a convex combination of the polygons~$P_k(\lambda_1)$ and~$P_k(\lambda_2)$. It is easy to check that if a line~$m$ intersects both~$P_k(\lambda_1)$ and~$P_k(\lambda_2)$, then the line~$m$ must also intersect the convex combination~$P_k(\lambda_3)$.

Now assume that both~$D_{k,V}(\lambda_1)$ and~$D_{k,V}(\lambda_2)$ are true. Then for any median line~$m$ both~$P_k(\lambda_1)$ and~$P_k(\lambda_2)$ intersect~$m$. By the observation above, the convex combination~$P_k(\lambda_3)$ must also intersects~$m$. Repeating this fact for all median lines implies that~$P_k(\lambda_3)$ intersects all median lines of~$V$. So~$D_{k,V}(\lambda_3)$ is true whenever~$D_{k,V}(\lambda_1)$ and~$D_{k,V}(\lambda_2)$ are true. Therefore, the set of parameters~$\{(r,x,y) \subseteq \mathbb R^3: D_{k,V}(r,x,y)\}$ is convex.
\end{proof}

\begin{lemma}
Every comparison in the decision algorithm in Theorem~1 either (i) does not depend on $(r,x,y)$, or (ii) has an associated critical hyperplane in $\mathbb R^3$.
\end{lemma}

\begin{proof}
Theorem~\ref{theorem:decide_pk} consists of three steps, computing the points outside the polygon, computing the event order, and processing the events. For the first two steps, the comparisons do depend on $(r,x,y)$ and have associated critical hyperplanes. We defer the proof of this claim to Sections~\ref{sec:subroutine_1} and~\ref{lemma:subroutine_2} respectively. For the third step, the comparisons do not depend on $(r,x,y)$ but rather the event order, so there is no requirement that comparisons have critical hyperplanes. 
\end{proof}

\begin{lemma}
The decision algorithm in Theorem~1 has an efficient parallel algorithm that runs on~$O(n)$ processors and takes~$O(\log n \cdot \log k)$ parallel steps per processor.
\end{lemma}

\begin{proof}
Given~$O(n)$ processors, we decide which points are outside the polygon in parallel by assiging a processor to each point. By Subroutine~1, this takes~$O(\log k)$ parallel steps per processor. We compute the event order in parallel using Preparata's sorting scheme~\cite{preparata1978new}. Each processor requires~$O(\log n)$ calls to Subroutine~2, so it total, each processor requires $O(\log n \cdot \log k)$ parallel steps. Finally, processing the events generates no critical hyperplanes, so this step does not require parallelisation.
\end{proof}

Now we combine Properties~1-3 with Megiddo's result to obtain an optimisation algorithm for the smallest, regular, $k$-sided polygon~$P_k(r,x,y)$ that intersects all median lines.

\begin{theorem}
\label{theorem:optimise_pk}
Given a set~$V$ of~$n$ points in the plane, there is an~$O(n \log^7 n \cdot \log^4 k)$ time algorithm to compute the minimum~$r$ such that~$D_{k,V}(r,x,y)$ is true for some regular,~$k$-sided polygon~$P_k(r,x,y)$.
\end{theorem}

\begin{proof}
Megiddo's multidimensional parametric search implies that there is an efficient optimisation algorithm. It only remains to show the running time of the technique.

The parallel algorithm runs on $P=O(n)$ processors in~$T_p = O(\log n \cdot \log k)$ parallel steps, whereas the decision algorithm runs in $T_s = O(n \log n \cdot \log k)$ time. The dimension~$d$ of the parameter space is three. The running time of multidimensional parametric search is~$O(T_p P + T_s (T_p \log P)^d)$ ~\cite{agarwal1998efficient}. Substituting our values into the above formula yields the required bound.
\end{proof}

\section{Computing Critical Hyperplanes}

The only requirement left to check is Property~2 for the comparisons in the comparison-based subroutines, that is, Subroutine~1 and Subroutine~2. Before launching into the analysis of the two subroutines, we first prove a tool. We will use the tool repeatedly in the next two sections to simplify checking Property~2.

\begin{lemma}
\label{lemma:get_crit}
Let gradient~$g \in \mathbb R$, point~$p \in \mathbb R^2$ and vector~$v \in \mathbb R^2$ be given, and let~$(r,x,y) \in \mathbb R^3$ be a variable parameter. Let~$L_{g,v}(r,x,y)$ be the line of gradient~$g$ through the point~$(x,y) + r \cdot v$. Then~$p$ is above, on, or below~$L_{g,v}(r,x,y)$ if and only if the point~$(r,x,y)$ is above, on, or below its associated critical hyperplane~$H_{p,g,v}$.
\end{lemma}

\begin{figure}[ht]
	\centering
	\begin{minipage}{0.4\textwidth}
		\begin{tikzpicture}[line cap=round,line join=round,>=triangle 45,x=3.0cm,y=3.0cm]
		\clip(0,0) rectangle (1,1);
		\draw [line width=1.pt,domain=0:1,color=bbbbbb] plot(\x,{(--0+0.8*\x)});
		\begin{scriptsize}
		\draw [fill=black] (0.3,0.7) circle (1pt);
		\draw[color=black] (0.25,0.75) node {$p$};
		\draw[color=bbbbbb] (0.7,0.3) node {$L_{g,v}(r,x,y)$};
		\end{scriptsize}
		\end{tikzpicture}
	\end{minipage}
	\begin{minipage}{0.05\textwidth}
		$\iff$
	\end{minipage}
	\begin{minipage}{0.4\textwidth}
		\begin{tikzpicture}[line cap=round,line join=round,>=triangle 45,x=3.0cm,y=3.0cm]
		\clip(0,0) rectangle (1,1);
		\fill[line width=1.pt,color=rrrrrr,fill=rrrrrr,fill opacity=0.2] (0.1,0.1) -- (0.35,0.6) -- (0.9,0.9) -- (0.6,0.35) -- cycle;
		\draw [line width=1.pt,color=rrrrrr] (0.1,0.1) -- (0.35,0.6);
		\draw [line width=1.pt,color=rrrrrr] (0.35,0.6) -- (0.9,0.9);
		\draw [line width=1.pt,color=rrrrrr] (0.9,0.9) -- (0.6,0.35);
		\draw [line width=1.pt,color=rrrrrr] (0.6,0.35) -- (0.1,0.1);
		\begin{scriptsize}
		\draw [fill=black] (0.3,0.7) circle (1pt);
		\draw[color=black] (0.25,0.75) node {$(r,x,y)$};
		\draw[color=rrrrrr] (0.5,0.5) node {$H_{p,g,v}$};
		\end{scriptsize}
		\end{tikzpicture}
	\end{minipage}
	\caption{Point~$p$ is above~$L_{g,v}(r,x,y)$ if and only if parameter~$(r,x,y)$ is above~$H_{p,g,v}$.}
	\label{fig:Lgv}
\end{figure}
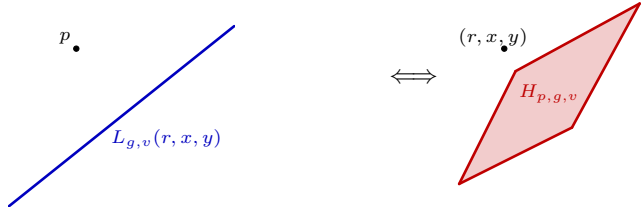

\begin{proof}
Let point~$p = (p_x, p_y)$ and vector~$v=(v_x,v_y)$. Now,~$(p_x,p_y)$ is above the line through~$(q_x, q_y)$ of gradient~$g$ if~$(p_x - q_x) - g \cdot (p_y - q_y) > 0$. Substituting the point~$(q_x, q_y) = (x,y) + r \cdot (v_1, v_2)$, we get the inequality
$$(p_2 - y - rv_2) - g \cdot (p_2 - x - rv_2) > 0.$$
This inequality can be rearranged into the form~$ax + by + cr + d > 0$, where
$$a = g,\ b=-1,\ c=(gv_1 - v_2),\ d = p_2 - gp_1.$$
In this form, we can see that the inequality is satisfied if and only if~$(r,x,y)$ lies above the hyperplane~$H_{p,g,v} :=  (ax + by + cr + d = 0)$, where~$a,b,c,d$ are given above. Hence, the two conditions,~$p$ above a line and~$(r,x,y)$ above a hyperplane, can be decided with the same inequality, which completes the proof.
\end{proof}

Now we are ready to address the subroutines.

\section{Subroutine~1}
\label{sec:subroutine_1}

Subroutine~1 decides whether a given point~$p$ is outside the~$k$-sided, regular polygon~$P_k(r,x,y)$. We present an~$O(\log k)$ time comparison-based algorithm and show that Property~2 holds.

\begin{lemma}
\label{lemma:subroutine_1}
Subroutine~1 has an~$O(\log k)$ time comparison-based algorithm, and comparisons in the algorithm that depend on the parameter~$(r,x,y)$ each have an associated critical hyperplane. 
\end{lemma}

\begin{proof}
We partition the polygon~$P_k(r,x,y)$ into~$k$ triangles, and decide which partition the point~$p$ is in, if it indeed is in any of these partitions. For~$1 \leq i \leq k$, the~$i^{th}$ partition of~$P_k(r,x,y)$ is the triangle joining the~$i^{th}$ vertex, the~$(i+1)^{th}$ vertex and the center of~$P_k(r,x,y)$. Figure~\ref{fig:i_i+1} shows the~$i^{th}$ partition of~$P_k(r,x,y)$.

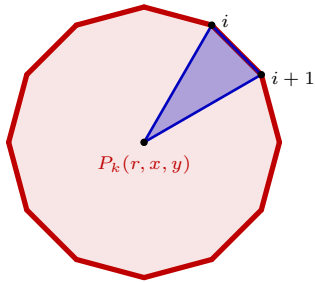
\begin{figure}[ht]
	\centering
	\begin{tikzpicture}[line cap=round,line join=round,>=triangle 45,x=1.0cm,y=1.0cm,scale=0.6]
	\clip(-4,-3.1) rectangle (4,4);
	\fill[line width=2.pt,color=rrrrrr,fill=rrrrrr,fill opacity=0.10000000149011612] (0.,3.) -- (-1.5,2.598076211353316) -- (-2.598076211353316,1.5) -- (-3.,0.) -- (-2.598076211353316,-1.5) -- (-1.5,-2.5980762113533156) -- (0.,-3.) -- (1.5,-2.5980762113533165) -- (2.5980762113533156,-1.5) -- (3.,0.) -- (2.598076211353317,1.5) -- (1.5,2.598076211353315) -- cycle;
	\draw [line width=2.pt,color=rrrrrr] (0.,3.)-- (-1.5,2.598076211353316);
	\draw [line width=2.pt,color=rrrrrr] (-1.5,2.598076211353316)-- (-2.598076211353316,1.5);
	\draw [line width=2.pt,color=rrrrrr] (-2.598076211353316,1.5)-- (-3.,0.);
	\draw [line width=2.pt,color=rrrrrr] (-3.,0.)-- (-2.598076211353316,-1.5);
	\draw [line width=2.pt,color=rrrrrr] (-2.598076211353316,-1.5)-- (-1.5,-2.5980762113533156);
	\draw [line width=2.pt,color=rrrrrr] (-1.5,-2.5980762113533156)-- (0.,-3.);
	\draw [line width=2.pt,color=rrrrrr] (0.,-3.)-- (1.5,-2.5980762113533165);
	\draw [line width=2.pt,color=rrrrrr] (1.5,-2.5980762113533165)-- (2.5980762113533156,-1.5);
	\draw [line width=2.pt,color=rrrrrr] (2.5980762113533156,-1.5)-- (3.,0.);
	\draw [line width=2.pt,color=rrrrrr] (3.,0.)-- (2.598076211353317,1.5);
	\draw [line width=2.pt,color=rrrrrr] (2.598076211353317,1.5)-- (1.5,2.598076211353315);
	\draw [line width=2.pt,color=rrrrrr] (1.5,2.598076211353315)-- (0.,3.);
	\fill[line width=2.pt,color=bbbbbb,fill=bbbbbb,fill opacity=0.3] (1.5,2.5980762113533147)-- (2.598076211353317,1.5) -- (0,0) -- cycle;
	\draw [line width=1.pt,color=bbbbbb] (1.5,2.5980762113533147)-- (2.598076211353317,1.5);
	\draw [line width=1.pt,color=bbbbbb] (0,0)-- (2.598076211353317,1.5);
	\draw [line width=1.pt,color=bbbbbb] (1.5,2.5980762113533147)-- (0,0);
	\begin{scriptsize}
	\draw [fill=black] (0.,0.) circle (2.0pt);
	\draw[color=rrrrrr] (0,-0.5) node {$P_k(r,x,y)$};
	\draw [fill=black] (1.5,2.5980762113533147) circle (2.0pt);
	\draw[color=black] (1.8,2.7) node {$i$};
	\draw [fill=black] (2.598076211353317,1.5) circle (2.0pt);
	\draw[color=black] (3.3,1.4) node {$i+1$};
	\end{scriptsize}
	\end{tikzpicture}	
	\caption{The~$i^{th}$ partition of~$P_k(r,x,y)$.}
	\label{fig:i_i+1}
\end{figure}

Assume for now that the point~$p$ is indeed in the polygon~$P_k(r,x,y)$ and hence in one of the~$k$ partitions. We decide whether~$p$ is in the $i^{th}$ partition for some $i \leq j$, or for some $i > j$, and perform a binary search for the index~$i$. This can be done by deciding if the point~$p$ is above, on, or below the line joining the center of~$P_k(r,x,y)$ and its~$j^{th}$ vertex. The comparison depends on~$(r,x,y)$, so we must compute its associated critical hyperplane using Lemma~\ref{lemma:get_crit}. Let $P_k(1,0,0)$ be the $k$-sided polygon of radius~1 and centered at the origin. Then set $g$ to be the gradient of the line joining the center to the~$i^{th}$ vertex of $P_k(1,0,0)$, and vector~$v = 0$ in Lemma~\ref{lemma:get_crit} to obtain the associated critical hyperplane.

We have searched for the partition that~$p$ is in if it is indeed in~$P_k(r,x,y)$. Hence, it only remains to decide whether~$p$ is indeed in that partition. This requires a constant number of comparisons, each of which depend on~$(r,x,y)$. We have already computed associated critical hyperplanes for two of the sides. The last side joins two adjacent vertices of the polygon~$P_k(r,x,y)$. Set~$g$ to be the gradient of the~$i^{th}$ side of polygon~$P_k(1,0,0)$, and the vector~$v$ to be the~$i^{th}$ vertex of $P_k(1,0,0)$, to obtain the final associated critical hyperplane.

The running time is dominated by the binary search for the~$i^{th}$ partition, which takes~$O(\log k)$ time.
\end{proof}

\section{Subroutine~2}
\label{sec:subroutine_2}

Subroutine~2 computes the relative clockwise order of four tangent lines drawn from two points to polygon~$P_k(r,x,y)$.

\begin{lemma}
\label{lemma:subroutine_2}
Subroutine~2 has an~$O(\log k)$-time comparison-based algorithm, and comparisons in the algorithm that depend on the parameter~$(r,x,y)$ each have an associated critical hyperplane.
\end{lemma}

\begin{proof}
Draw two lines~$t_i, t_j$ tangent to~$P_k(r,x,y)$ and parallel to~$pq$, and let the points of tangency be vertex~$i$ and vertex~$j$. If there are multiple points of tangency then choose any such point. Then without loss of generality, set~$ij$ to be horizontal, and assume further that~$p$ has a larger~$y$ coordinate than~$q$. Then the~$t_i, t_j$ and~$ij$ partition the plane into the four regions, as shown in Figure~\ref{fig:udlr}. Region~$L$ is left of both tangents,~$R$ is right of both tangents,~$U$ is between the tangents and above~$ij$, and~$D$ is between the tangents and below~$ij$.

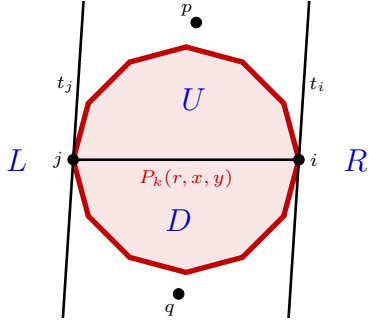
\begin{figure}[ht]
	\centering
	\begin{tikzpicture}[line cap=round,line join=round,>=triangle 45,x=0.5cm,y=0.5cm]
	\clip(-5,-4.2) rectangle (5,4.2);
	\fill[line width=2.pt,color=rrrrrr,fill=rrrrrr,fill opacity=0.10000000149011612] (0.,3.) -- (-1.5,2.598076211353316) -- (-2.598076211353316,1.5) -- (-3.,0.) -- (-2.598076211353316,-1.5) -- (-1.5,-2.5980762113533156) -- (0.,-3.) -- (1.5,-2.5980762113533165) -- (2.5980762113533156,-1.5) -- (3.,0.) -- (2.598076211353317,1.5) -- (1.5,2.598076211353315) -- cycle;
	\draw [line width=2.pt,color=rrrrrr] (0.,3.)-- (-1.5,2.598076211353316);
	\draw [line width=2.pt,color=rrrrrr] (-1.5,2.598076211353316)-- (-2.598076211353316,1.5);
	\draw [line width=2.pt,color=rrrrrr] (-2.598076211353316,1.5)-- (-3.,0.);
	\draw [line width=2.pt,color=rrrrrr] (-3.,0.)-- (-2.598076211353316,-1.5);
	\draw [line width=2.pt,color=rrrrrr] (-2.598076211353316,-1.5)-- (-1.5,-2.5980762113533156);
	\draw [line width=2.pt,color=rrrrrr] (-1.5,-2.5980762113533156)-- (0.,-3.);
	\draw [line width=2.pt,color=rrrrrr] (0.,-3.)-- (1.5,-2.5980762113533165);
	\draw [line width=2.pt,color=rrrrrr] (1.5,-2.5980762113533165)-- (2.5980762113533156,-1.5);
	\draw [line width=2.pt,color=rrrrrr] (2.5980762113533156,-1.5)-- (3.,0.);
	\draw [line width=2.pt,color=rrrrrr] (3.,0.)-- (2.598076211353317,1.5);
	\draw [line width=2.pt,color=rrrrrr] (2.598076211353317,1.5)-- (1.5,2.598076211353315);
	\draw [line width=2.pt,color=rrrrrr] (1.5,2.598076211353315)-- (0.,3.);
	\draw [line width=1.pt,color=black] (-3.,0.)-- (3.,0.);
	\draw [line width=1.pt,color=black,domain=-7.73603361203088:32.919135820740756] plot(\x,{(-30.390535595770043-10.130178531923347*\x)/-0.658965175256351});
	\draw [line width=1.pt,color=black,domain=-7.73603361203088:32.919135820740756] plot(\x,{(--36.644037616961384-12.214679205653795*\x)/-0.7945613394757755});
	\begin{scriptsize}
	\draw[color=rrrrrr] (0,-0.5) node {$P_k(r,x,y)$};
	\draw[color=bbbbbb] (0.2,1.6) node {\large $U$};
	\draw[color=bbbbbb] (-0.2,-1.6) node {\large $D$};
	\draw[color=bbbbbb] (-4.5,0) node {\large $L$};
	\draw[color=bbbbbb] (4.5,0) node {\large $R$};
	\draw[color=black] (3.5,2) node {$t_i$};
	\draw[color=black] (-3.2,2) node {$t_j$};
	\draw [fill=black] (3,0) circle (2.0pt);
	\draw[color=black] (3.4,0) node {$i$};
	\draw [fill=black] (-3,0) circle (2.0pt);
	\draw[color=black] (-3.4,0) node {$j$};
	\draw [fill=black] (0.2727377160866591,3.6505425902824435) circle (2.0pt);
	\draw[color=black] (0.01927714469451474,4.0) node {$p$};
	\draw [fill=black] (-0.19702014771574028,-3.570978676364149) circle (2.0pt);
	\draw[color=black] (-0.4383343595843128,-4.0) node {$q$};
	\end{scriptsize}
	\end{tikzpicture}
	\caption{The lines~$t_i$,~$t_j$,~$ij$ partition the plane into regions~$L, R, U, D$.}
	\label{fig:udlr}
\end{figure}

Then the relative clockwise order of the four lines drawn from~$p$ and~$q$ are determined by which of the four regions~$L$,~$R$,~$U$ or~$D$ the points~$p$ and~$q$ are located. 
See Figure~\ref{fig:pqpq}. 

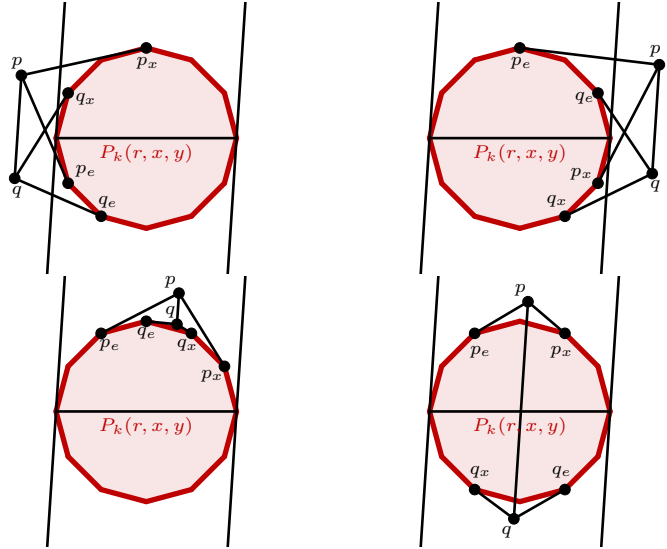
\begin{figure}[ht]
	\centering
	\begin{minipage}{0.4\textwidth}
		\begin{tikzpicture}[line cap=round,line join=round,>=triangle 45,x=0.4cm,y=0.4cm]
		\clip(-5,-4.5) rectangle (5,4.5);
		\fill[line width=2.pt,color=rrrrrr,fill=rrrrrr,fill opacity=0.10000000149011612] (0.,3.) -- (-1.5,2.598076211353316) -- (-2.598076211353316,1.5) -- (-3.,0.) -- (-2.598076211353316,-1.5) -- (-1.5,-2.5980762113533156) -- (0.,-3.) -- (1.5,-2.5980762113533165) -- (2.5980762113533156,-1.5) -- (3.,0.) -- (2.598076211353317,1.5) -- (1.5,2.598076211353315) -- cycle;
		\draw [line width=2.pt,color=rrrrrr] (0.,3.)-- (-1.5,2.598076211353316);
		\draw [line width=2.pt,color=rrrrrr] (-1.5,2.598076211353316)-- (-2.598076211353316,1.5);
		\draw [line width=2.pt,color=rrrrrr] (-2.598076211353316,1.5)-- (-3.,0.);
		\draw [line width=2.pt,color=rrrrrr] (-3.,0.)-- (-2.598076211353316,-1.5);
		\draw [line width=2.pt,color=rrrrrr] (-2.598076211353316,-1.5)-- (-1.5,-2.5980762113533156);
		\draw [line width=2.pt,color=rrrrrr] (-1.5,-2.5980762113533156)-- (0.,-3.);
		\draw [line width=2.pt,color=rrrrrr] (0.,-3.)-- (1.5,-2.5980762113533165);
		\draw [line width=2.pt,color=rrrrrr] (1.5,-2.5980762113533165)-- (2.5980762113533156,-1.5);
		\draw [line width=2.pt,color=rrrrrr] (2.5980762113533156,-1.5)-- (3.,0.);
		\draw [line width=2.pt,color=rrrrrr] (3.,0.)-- (2.598076211353317,1.5);
		\draw [line width=2.pt,color=rrrrrr] (2.598076211353317,1.5)-- (1.5,2.598076211353315);
		\draw [line width=2.pt,color=rrrrrr] (1.5,2.598076211353315)-- (0.,3.);
		\draw [line width=1.pt,color=black] (-3.,0.)-- (3.,0.);
		\draw [line width=1.pt,color=black,domain=-7.73603361203088:32.919135820740756] plot(\x,{(-30.390535595770043-10.130178531923347*\x)/-0.658965175256351});
		\draw [line width=1.pt,color=black,domain=-7.73603361203088:32.919135820740756] plot(\x,{(--36.644037616961384-12.214679205653795*\x)/-0.7945613394757755});
		\draw [line width=1.pt,color=black] (-4.373433184044156,-1.3367240527387239)-- (-2.598076211353317,1.5);
		\draw [line width=1.pt,color=black] (-2.598076211353316,-1.5)-- (-4.150604247123473,2.088793710331325);
		\draw [line width=1.pt,color=black] (-4.373433184044156,-1.3367240527387239)-- (-1.5,-2.5980762113533156);
		\draw [line width=1.pt,color=black] (-4.150604247123473,2.088793710331325)-- (0.,3.);
		\draw [line width=1.pt,color=black] (-4.150604247123473,2.088793710331325)-- (-4.373433184044156,-1.3367240527387239);
		\begin{scriptsize}
		\draw[color=rrrrrr] (0,-0.5) node {$P_k(r,x,y)$};
		\draw [fill=black] (-4.150604247123473,2.088793710331325) circle (2.0pt);
		\draw[color=black] (-4.3,2.5) node {$p$};
		\draw [fill=black] (-4.373433184044156,-1.3367240527387239) circle (2.0pt);
		\draw[color=black] (-4.3,-1.8) node {$q$};
		\draw [fill=black] (-2.598076211353317,1.5) circle (2.0pt);
		\draw[color=black] (-2.0,1.3) node {$q_x$};
		\draw [fill=black] (-2.598076211353316,-1.5) circle (2.0pt);
		\draw[color=black] (-2.0,-1.1) node {$p_e$};
		\draw [fill=black] (0.,3.) circle (2.0pt);
		\draw[color=black] (0.055404368716527444,2.5) node {$p_x$};
		\draw [fill=black] (-1.5,-2.5980762113533156) circle (2.0pt);
		\draw[color=black] (-1.269260512090605,-2.128992819546285) node {$q_e$};
		\end{scriptsize}
		\end{tikzpicture}
	\end{minipage}
	\begin{minipage}{0.4\textwidth}
		\begin{tikzpicture}[line cap=round,line join=round,>=triangle 45,x=0.4cm,y=0.4cm]
		\clip(-5,-4.5) rectangle (5,4.5);
		\fill[line width=2.pt,color=rrrrrr,fill=rrrrrr,fill opacity=0.10000000149011612] (0.,3.) -- (-1.5,2.598076211353316) -- (-2.598076211353316,1.5) -- (-3.,0.) -- (-2.598076211353316,-1.5) -- (-1.5,-2.5980762113533156) -- (0.,-3.) -- (1.5,-2.5980762113533165) -- (2.5980762113533156,-1.5) -- (3.,0.) -- (2.598076211353317,1.5) -- (1.5,2.598076211353315) -- cycle;
		\draw [line width=2.pt,color=rrrrrr] (0.,3.)-- (-1.5,2.598076211353316);
		\draw [line width=2.pt,color=rrrrrr] (-1.5,2.598076211353316)-- (-2.598076211353316,1.5);
		\draw [line width=2.pt,color=rrrrrr] (-2.598076211353316,1.5)-- (-3.,0.);
		\draw [line width=2.pt,color=rrrrrr] (-3.,0.)-- (-2.598076211353316,-1.5);
		\draw [line width=2.pt,color=rrrrrr] (-2.598076211353316,-1.5)-- (-1.5,-2.5980762113533156);
		\draw [line width=2.pt,color=rrrrrr] (-1.5,-2.5980762113533156)-- (0.,-3.);
		\draw [line width=2.pt,color=rrrrrr] (0.,-3.)-- (1.5,-2.5980762113533165);
		\draw [line width=2.pt,color=rrrrrr] (1.5,-2.5980762113533165)-- (2.5980762113533156,-1.5);
		\draw [line width=2.pt,color=rrrrrr] (2.5980762113533156,-1.5)-- (3.,0.);
		\draw [line width=2.pt,color=rrrrrr] (3.,0.)-- (2.598076211353317,1.5);
		\draw [line width=2.pt,color=rrrrrr] (2.598076211353317,1.5)-- (1.5,2.598076211353315);
		\draw [line width=2.pt,color=rrrrrr] (1.5,2.598076211353315)-- (0.,3.);
		\draw [line width=1.pt,color=black] (-3.,0.)-- (3.,0.);
		\draw [line width=1.pt,color=black,domain=-7.73603361203088:32.919135820740756] plot(\x,{(-30.390535595770043-10.130178531923347*\x)/-0.658965175256351});
		\draw [line width=1.pt,color=black,domain=-7.73603361203088:32.919135820740756] plot(\x,{(--36.644037616961384-12.214679205653795*\x)/-0.7945613394757755});
		\draw [line width=1.pt] (4.638213864457156,2.4412569415606393)-- (0.,3.);
		\draw [line width=1.pt,color=black] (4.638213864457156,2.4412569415606393)-- (4.402649476742617,-1.1800411071308936);
		\draw [line width=1.pt] (4.638213864457156,2.4412569415606393)-- (2.598076211353315,-1.5);
		\draw [line width=1.pt] (4.402649476742617,-1.1800411071308936)-- (2.598076211353317,1.5);
		\draw [line width=1.pt] (4.402649476742617,-1.1800411071308936)-- (1.5,-2.598076211353316);
		\begin{scriptsize}
		\draw[color=rrrrrr] (0,-0.5) node {$P_k(r,x,y)$};
		\draw [fill=black] (4.638213864457156,2.4412569415606393) circle (2pt);
		\draw[color=black] (4.5,2.9) node {$p$};
		\draw [fill=black] (4.402649476742617,-1.1800411071308936) circle (2pt);
		\draw[color=black] (4.5,-1.6) node {$q$};
		\draw [fill=black] (0.,3.) circle (2.0pt);
		\draw[color=black] (0.055404368716527444,2.5) node {$p_e$};
		\draw [fill=black] (2.598076211353315,-1.5) circle (2.0pt);
		\draw[color=black] (2.0664865059419015,-1.1776425869666178) node {$p_x$};
		\draw [fill=black] (2.598076211353317,1.5) circle (2.0pt);
		\draw[color=black] (2.162825770000602,1.3512630945742692) node {$q_e$};
		\draw [fill=black] (1.5,-2.598076211353316) circle (2.0pt);
		\draw[color=black] (1.2476027614429468,-2.0687807795095967) node {$q_x$};
		\end{scriptsize}
		\end{tikzpicture}
	\end{minipage}
	\begin{minipage}{0.4\textwidth}
		\begin{tikzpicture}[line cap=round,line join=round,>=triangle 45,x=0.4cm,y=0.4cm]
		\clip(-5,-4.5) rectangle (5,4.5);
		\fill[line width=2.pt,color=rrrrrr,fill=rrrrrr,fill opacity=0.10000000149011612] (0.,3.) -- (-1.5,2.598076211353316) -- (-2.598076211353316,1.5) -- (-3.,0.) -- (-2.598076211353316,-1.5) -- (-1.5,-2.5980762113533156) -- (0.,-3.) -- (1.5,-2.5980762113533165) -- (2.5980762113533156,-1.5) -- (3.,0.) -- (2.598076211353317,1.5) -- (1.5,2.598076211353315) -- cycle;
		\draw [line width=2.pt,color=rrrrrr] (0.,3.)-- (-1.5,2.598076211353316);
		\draw [line width=2.pt,color=rrrrrr] (-1.5,2.598076211353316)-- (-2.598076211353316,1.5);
		\draw [line width=2.pt,color=rrrrrr] (-2.598076211353316,1.5)-- (-3.,0.);
		\draw [line width=2.pt,color=rrrrrr] (-3.,0.)-- (-2.598076211353316,-1.5);
		\draw [line width=2.pt,color=rrrrrr] (-2.598076211353316,-1.5)-- (-1.5,-2.5980762113533156);
		\draw [line width=2.pt,color=rrrrrr] (-1.5,-2.5980762113533156)-- (0.,-3.);
		\draw [line width=2.pt,color=rrrrrr] (0.,-3.)-- (1.5,-2.5980762113533165);
		\draw [line width=2.pt,color=rrrrrr] (1.5,-2.5980762113533165)-- (2.5980762113533156,-1.5);
		\draw [line width=2.pt,color=rrrrrr] (2.5980762113533156,-1.5)-- (3.,0.);
		\draw [line width=2.pt,color=rrrrrr] (3.,0.)-- (2.598076211353317,1.5);
		\draw [line width=2.pt,color=rrrrrr] (2.598076211353317,1.5)-- (1.5,2.598076211353315);
		\draw [line width=2.pt,color=rrrrrr] (1.5,2.598076211353315)-- (0.,3.);
		\draw [line width=1.pt,color=black] (-3.,0.)-- (3.,0.);
		\draw [line width=1.pt,color=black,domain=-3.8428496566681756:13.568405208833058] plot(\x,{(-30.390535595770043-10.130178531923347*\x)/-0.658965175256351});
		\draw [line width=1.pt,color=black,domain=-3.8428496566681756:13.568405208833058] plot(\x,{(--36.644037616961384-12.214679205653795*\x)/-0.7945613394757755});
		\draw [line width=1.pt,color=black] (1.086865649136689,3.935036751215582)-- (1.019560032015581,2.900357009038221);
		\draw [line width=1.pt,color=black] (1.086865649136689,3.935036751215582)-- (2.598076211353317,1.5);
		\draw [line width=1.pt,color=black] (1.019560032015581,2.900357009038221)-- (1.5,2.5980762113533147);
		\draw [line width=1.pt,color=black] (1.019560032015581,2.900357009038221)-- (0.,3.);
		\draw [line width=1.pt,color=black] (1.086865649136689,3.935036751215582)-- (-1.5,2.5980762113533165);
		\begin{scriptsize}
		\draw[color=rrrrrr] (0,-0.5) node {$P_k(r,x,y)$};
		\draw [fill=black] (1.086865649136689,3.935036751215582) circle (2.0pt);
		\draw[color=black] (0.8,4.3) node {$p$};
		\draw [fill=black] (1.019560032015581,2.900357009038221) circle (2.0pt);
		\draw[color=black] (0.8,3.3) node {$q$};
		\draw [fill=black] (2.598076211353317,1.5) circle (2.0pt);
		\draw[color=black] (2.2,1.2) node {$p_x$};
		\draw [fill=black] (-1.5,2.5980762113533165) circle (2.0pt);
		\draw[color=black] (-1.2,2.2) node {$p_e$};
		\draw [fill=black] (1.5,2.5980762113533147) circle (2.0pt);
		\draw[color=black] (1.3761133036500017,2.2) node {$q_x$};
		\draw [fill=black] (0.,3.) circle (2.0pt);
		\draw[color=black] (0.018962259000666032,2.6) node {$q_e$};
		\end{scriptsize}
		\end{tikzpicture}
	\end{minipage}
	\begin{minipage}{0.4\textwidth}
		\begin{tikzpicture}[line cap=round,line join=round,>=triangle 45,x=0.4cm,y=0.4cm]
		\clip(-5,-4.5) rectangle (5,4.5);
		\fill[line width=2.pt,color=rrrrrr,fill=rrrrrr,fill opacity=0.10000000149011612] (0.,3.) -- (-1.5,2.598076211353316) -- (-2.598076211353316,1.5) -- (-3.,0.) -- (-2.598076211353316,-1.5) -- (-1.5,-2.5980762113533156) -- (0.,-3.) -- (1.5,-2.5980762113533165) -- (2.5980762113533156,-1.5) -- (3.,0.) -- (2.598076211353317,1.5) -- (1.5,2.598076211353315) -- cycle;
		\draw [line width=2.pt,color=rrrrrr] (0.,3.)-- (-1.5,2.598076211353316);
		\draw [line width=2.pt,color=rrrrrr] (-1.5,2.598076211353316)-- (-2.598076211353316,1.5);
		\draw [line width=2.pt,color=rrrrrr] (-2.598076211353316,1.5)-- (-3.,0.);
		\draw [line width=2.pt,color=rrrrrr] (-3.,0.)-- (-2.598076211353316,-1.5);
		\draw [line width=2.pt,color=rrrrrr] (-2.598076211353316,-1.5)-- (-1.5,-2.5980762113533156);
		\draw [line width=2.pt,color=rrrrrr] (-1.5,-2.5980762113533156)-- (0.,-3.);
		\draw [line width=2.pt,color=rrrrrr] (0.,-3.)-- (1.5,-2.5980762113533165);
		\draw [line width=2.pt,color=rrrrrr] (1.5,-2.5980762113533165)-- (2.5980762113533156,-1.5);
		\draw [line width=2.pt,color=rrrrrr] (2.5980762113533156,-1.5)-- (3.,0.);
		\draw [line width=2.pt,color=rrrrrr] (3.,0.)-- (2.598076211353317,1.5);
		\draw [line width=2.pt,color=rrrrrr] (2.598076211353317,1.5)-- (1.5,2.598076211353315);
		\draw [line width=2.pt,color=rrrrrr] (1.5,2.598076211353315)-- (0.,3.);
		\draw [line width=1.pt,color=black] (-3.,0.)-- (3.,0.);
		\draw [line width=1.pt,color=black,domain=-7.73603361203088:32.919135820740756] plot(\x,{(-30.390535595770043-10.130178531923347*\x)/-0.658965175256351});
		\draw [line width=1.pt,color=black,domain=-7.73603361203088:32.919135820740756] plot(\x,{(--36.644037616961384-12.214679205653795*\x)/-0.7945613394757755});
		\draw [line width=1.pt,color=black] (0.2727377160866591,3.6505425902824435)-- (-0.19702014771574028,-3.570978676364149);
		\draw [line width=1.pt,color=black] (0.2727377160866591,3.6505425902824435)-- (-1.5,2.5980762113533165);
		\draw [line width=1.pt,color=black] (0.2727377160866591,3.6505425902824435)-- (1.5,2.5980762113533147);
		\draw [line width=1.pt,color=black] (-0.19702014771574028,-3.570978676364149)-- (-1.5,-2.5980762113533156);
		\draw [line width=1.pt,color=black] (-0.19702014771574028,-3.570978676364149)-- (1.5,-2.598076211353316);
		\begin{scriptsize}
		\draw[color=rrrrrr] (0,-0.5) node {$P_k(r,x,y)$};
		\draw [fill=black] (0.2727377160866591,3.6505425902824435) circle (2.0pt);
		\draw[color=black] (0.01927714469451474,4.048762488217882) node {$p$};
		\draw [fill=black] (-0.19702014771574028,-3.570978676364149) circle (2.0pt);
		\draw[color=black] (-0.4383343595843128,-4.0) node {$q$};
		\draw [fill=black] (-1.5,2.5980762113533165) circle (2.0pt);
		\draw[color=black] (-1.3,2) node {$p_e$};
		\draw [fill=black] (1.5,2.5980762113533147) circle (2.0pt);
		\draw[color=black] (1.4,2) node {$p_x$};
		\draw [fill=black] (-1.5,-2.5980762113533156) circle (2.0pt);
		\draw[color=black] (-1.3,-2) node {$q_x$};
		\draw [fill=black] (1.5,-2.598076211353316) circle (2.0pt);
		\draw[color=black] (1.4,-2) node {$q_e$};
		\end{scriptsize}
		\end{tikzpicture}

	\end{minipage}
	\caption{The relative orders shown for when~$(i)$~$p,q \in L$,~$(ii)$~$p,q \in R$,~$((iii)$~$p,q \in U$ and~$(iv)$~$p \in U, q \in D$.}
	\label{fig:pqpq}
\end{figure}

Five cases follows. Let~$p_e$ and~$p_x$ points of tangency from~$p$ such that the points~$p_e,p,p_x$ are in clockwise order. If~$p,q$ are in the same region, then the containing region~$L$,~$R$,~$U$, and~$D$ correspond to the relative clockwise orders~$q_e p_e q_x p_x$,~$p_e q_e p_x q_x$,~$p_e q_e q_x p_x$, and~$q_e p_e p_x q_x$ respectively. If~$p,q$ are in different regions, then they must be in~$U$ and~$D$ respectively, and the relative order is~$p_e p_x q_e q_x$. The proof for case analysis for the five cases is omitted, but the diagrams in Figure~\ref{fig:pqpq} may be useful for the reader.

The running time of the algorithm is as follows. Given the gradient of~$pq$, there is an~$O(\log k)$ time algorithm to binary search the gradients of the sides of~$P_k(r,x,y)$ to compute the vertices~$i$ and~$j$. Then the remainder of the algorithm takes constant time: rotating the diagram so that~$ij$ is horizontal, deciding whether~$p$ or~$q$ has a larger~$y$ coordinate, and computing the region~$L,R,U,D$ that points~$p,q$ are in.

The proof of existence of critical hyperplanes is as follows. Since the gradients of~$pq$ and sides of~$P_k$ do not depend on~$(r,x,y)$, computing~$i$ and~$j$ generates no critical hyperplanes. Similarly, rotating the diagram so that~$ij$ is horizontal and then deciding which of~$p$ or~$q$ have larger~$y$ coordinates also generates no critical hyperplanes. It only remains to decide which of the four regions~$L,R,U,D$ the point~$p$, and respectively~$q$, is in. Set~$g$ to the gradient of~$pq$ and vector~$v$ to be the~$i^{th}$ vertex of~$P_k(1,0,0)$ in Lemma~\ref{lemma:get_crit} to decide if~$p$ is to the left of the tangent through~$i$. Do so similarly for~$j$ to decide if~$p$ is to the right of the tangent through~$j$. Finally, set~$g$ to the gradient of~$ij$ and vector~$v$ to be either the~$i^{th}$ or~$j^{th}$ vertex of~$P_k(1,0,0)$ to decide if~$p$ is above the chord~$ij$.
\end{proof}

Checking that Property~2 holds for the comparison-based subroutines, Subroutine~1 and Subroutine~2, completes the proof to Theorem~2. In the final section we will prove that Theorem~2 implies that we have an efficient algorithm for computing the yolk in the $\mathcal L_1$ and $\mathcal L_\infty$ meetrics, and an efficient approximation algorithm for the $\mathcal L_2$ metric.

\section{Computing the Yolk in~$\mathcal L_1, \mathcal L_2$, and~$\mathcal L_\infty$}
\label{sec:corollary}
It remains to show that our general problem for~$P_k(r,x,y)$ implies the results as claimed in the introduction. 

\begin{theorem}
Given a set~$V$ of~$n$ points in the plane, there is an~$O(n \log^7 n)$ time algorithm to compute the yolk of~$V$ in the~$\mathcal L_1$ and~$\mathcal L_\infty$ metrics.
\end{theorem}
\begin{proof}
Setting~$k=4$ in Theorem~\ref{theorem:optimise_pk} gives an algorithm to compute the smallest~$P_4(r,x,y)$ that intersects all median lines of~$V$ in~$O(n \log^7 n)$ time. This rotated square coincides with yolk in the~$\mathcal L_1$ metric, refer to Figure~\ref{fig:l1_yolk} and Definition~1.

Computing the yolk in the~$\mathcal L_\infty$ metric requires one extra step. Rotate the points of~$V$ by~$45^\circ$ clockwise, compute the smallest~$P_4(r,x,y)$, and then rotate the square~$P_4(r,x,y)$ back~$45^\circ$ anticlockwise to obtain the yolk in the~$\mathcal L_\infty$ metric. 
\end{proof}

\begin{theorem}
Given a set~$V$ of~$n$ points in the plane and an~$\varepsilon > 0$, there is an~$O(n \log^7 n \cdot \log^4 \frac 1 \varepsilon)$ time algorithm to compute a~$(1+\varepsilon)$-approximation of the yolk in the~$\mathcal L_2$ metric.
\end{theorem}
\begin{proof}
Setting~$k \approx \pi \cdot (1+\frac 1 \varepsilon)$ in Theorem~\ref{theorem:optimise_pk} gives an algorithm to compute the smallest~$P_k(r,x,y)$ that intersects all median lines of~$V$ in the desired running time. It suffices to show that for this parameter set~$(r,x,y)$, the disk centered at~$(x,y)$ with radius~$r$ is a~$(1 + \varepsilon)$-approximation for the yolk in the~$\mathcal L_2$ metric.

First, note that~$P_k(r,x,y)$ intersects all median lines, and~$B(r,x,y)$ encloses~$P_k(r,x,y)$, so the disk must also intersect all median lines of~$V$. Hence, it suffices to show that the radius~$r$ of~$B(r,x,y)$ satisfies~$r \leq (1+\varepsilon) \cdot r_2$, where~$r_2$ is the radius of the true yolk in the~$\mathcal L_2$ metric.

Let the yolk in the~$\mathcal L_2$ metric be the disk~$B(r_2, x_2, y_2)$. Consider the regular,~$k$-sided polygon~$P_k(r_2 \cdot \sec \frac \pi k,x_2,y_2)$, so that by construction, all sides of this polygon are tangent to~$B(r_2,x_2,y_2)$.

\begin{figure}[ht]
	\centering
	\begin{tikzpicture}[line cap=round,line join=round,>=triangle 45,x=0.6cm,y=0.6cm]
	\clip(-3,-3) rectangle (3,3.1);
	\fill[line width=2.pt,color=rrrrrr,fill=rrrrrr,fill opacity=0.10000000149011612] (0.,3.) -- (-1.5,2.598076211353316) -- (-2.598076211353316,1.5) -- (-3.,0.) -- (-2.598076211353316,-1.5) -- (-1.5,-2.5980762113533156) -- (0.,-3.) -- (1.5,-2.5980762113533165) -- (2.5980762113533156,-1.5) -- (3.,0.) -- (2.598076211353317,1.5) -- (1.5,2.598076211353315) -- cycle;
	\draw [line width=2.pt,color=rrrrrr] (0.,3.)-- (-1.5,2.598076211353316);
	\draw [line width=2.pt,color=rrrrrr] (-1.5,2.598076211353316)-- (-2.598076211353316,1.5);
	\draw [line width=2.pt,color=rrrrrr] (-2.598076211353316,1.5)-- (-3.,0.);
	\draw [line width=2.pt,color=rrrrrr] (-3.,0.)-- (-2.598076211353316,-1.5);
	\draw [line width=2.pt,color=rrrrrr] (-2.598076211353316,-1.5)-- (-1.5,-2.5980762113533156);
	\draw [line width=2.pt,color=rrrrrr] (-1.5,-2.5980762113533156)-- (0.,-3.);
	\draw [line width=2.pt,color=rrrrrr] (0.,-3.)-- (1.5,-2.5980762113533165);
	\draw [line width=2.pt,color=rrrrrr] (1.5,-2.5980762113533165)-- (2.5980762113533156,-1.5);
	\draw [line width=2.pt,color=rrrrrr] (2.5980762113533156,-1.5)-- (3.,0.);
	\draw [line width=2.pt,color=rrrrrr] (3.,0.)-- (2.598076211353317,1.5);
	\draw [line width=2.pt,color=rrrrrr] (2.598076211353317,1.5)-- (1.5,2.598076211353315);
	\draw [line width=2.pt,color=rrrrrr] (1.5,2.598076211353315)-- (0.,3.);
	\draw [line width=1.pt, color=bbbbbb, fill opacity = 0.1] (0.,0.) circle (2.8977774788672046);
	\draw [line width=1.pt, color=rrrrrr] (0.,0.)-- (0.,3.);
	\draw [line width=1.pt, color=bbbbbb] (0.,0.)-- (0.75,2.799038105676657);
	\begin{scriptsize}
	\draw [fill=black] (0.,0.) circle (2.0pt);
	\draw[color=rrrrrr] (0,-1) node {$P_k(r_2\cdot \sec \frac \pi k,x_2,y_2)$};
	\draw[color=rrrrrr] (-1.1,1.5) node {$r_2\cdot \sec \frac \pi k$};
	\draw[color=bbbbbb] (0.7,1.4) node {$r$};
	\end{scriptsize}
	\end{tikzpicture}
	\caption{The polygon~$P_k(r_2 \cdot \sec \frac \pi k,x_2,y_2)$ is externally tangent to the disk~$B(r_2,x_2,y_2)$.}
	\label{fig:sec}
\end{figure}
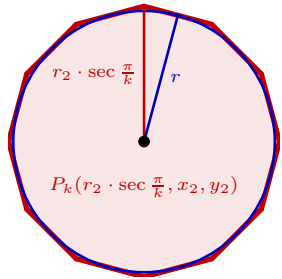

Now since~$B(r_2,x_2,y_2)$ is the~$\mathcal L_2$ yolk, it intersects all median lines and so does its enclosing polygon~$P_k(r_2 \cdot \sec \frac \pi k,x_2,y_2)$. By the minimality of~$P_k(r,x,y)$, we get~$r \leq \sec \frac \pi k \cdot r_2$. But for~$\theta \in [0,\frac \pi 3]$, we have~$\sec \theta \leq \frac 1 {1-\theta}$. So,
\[
	\sec \frac \pi k \leq \frac 1 { 1 - \frac \pi k} \leq 1+\varepsilon,
\]
which implies that~$r \leq (1+\varepsilon) \cdot r_2$, as required.
\end{proof}

\section{Concluding Remarks}

Cole's~\cite{cole1987slowing} extension to parametric search states that the running time of the parametric search may be reduced if certain comparisons are delayed. This is a direction for further research that could potentially improve the running time of our algorithms.

An open problem is whether one can compute the yolk in higher dimensions without precomputing all median hyperplanes. Avoiding the computation of median hyperplanes yields even greater benefits as less is known about bounds on the number of median hyperplanes in higher dimensions.

Similarly, our approximation algorithm for the $\mathcal L_2$~yolk in the plane is optimal up to polylogarithmic factors, however, it is an open problem as to whether there is a near-linear time exact algorithm. Our attempts to apply Megiddo's parametric search technique to the~$\mathcal L_2$~yolk have been unsuccessful so far.

Finally, there are other solution concepts in computational spatial voting that currently lack efficient algorithms. The shortcomings of existing algorithms are: for the Shapley Owen power score there is only an approximate algorithm~\cite{godfrey2005computation}, for the Finagle point only regular polygons have been considered~\cite{wuffle1989finagle} and for the~$\varepsilon$-core only a membership algorithm exists~\cite{tovey2010finite}. Since these problems have a close connection to either median lines or minimal radius, we suspect that Megiddo's parametric search technique could also be useful for these problems.

\bibliographystyle{plain}
\bibliography{AAAI19-ComputeYolk.bib}

\end{document}